\documentclass[letterpaper,11pt]{article}

\usepackage{amsmath}
\usepackage{amsfonts}
\usepackage{amssymb}
\usepackage{amsthm}
\usepackage{graphicx}
\usepackage[usenames]{color}
\usepackage{fullpage}
\usepackage{url,ifthen}
\usepackage{srcltx}
\usepackage{multirow}
\usepackage{boxedminipage}
\usepackage[margin=1.2in]{geometry}
\usepackage{nicefrac}
\usepackage{xspace}
\usepackage{wrapfig}

\definecolor{DarkGreen}{rgb}{0.1,0.5,0.1}
\definecolor{DarkRed}{rgb}{0.5,0.1,0.1}
\definecolor{DarkBlue}{rgb}{0.1,0.1,0.5}

\usepackage[pdftex]{hyperref}
\hypersetup{
    unicode=false,          
    pdftoolbar=true,        
    pdfmenubar=true,        
    pdffitwindow=false,      
    pdfnewwindow=true,      
    colorlinks=true,       
    linkcolor=DarkRed,          
    citecolor=DarkGreen,        
    filecolor=DarkRed,      
    urlcolor=DarkBlue,          
    pdftitle={},
    pdfauthor={},
}

\def\draft{1}   

\ifnum\draft=1 
    \def\ShowAuthNotes{1}
\else
    \def\ShowAuthNotes{0}
\fi

\ifnum\ShowAuthNotes=1
\newcommand{\authnote}[2]{{ \footnotesize \bf{\color{red}[#1's Note: {\color{blue}#2}]}}}
\else
\newcommand{\authnote}[2]{}
\fi

\newtheorem{theorem}{Theorem}[section]
\newtheorem{open}{Open Question}[section]
\newtheorem{corollary}[theorem]{Corollary}

\newtheorem{lemma}[theorem]{Lemma}
\newtheorem{proposition}{Proposition}[section]
\newtheorem{claim}[theorem]{Claim}
\newtheorem{fact}{Fact}[section]
\theoremstyle{definition}
\newtheorem{definition}{Definition}[section]
\newtheorem{remark}{Remark}[section]

\newcommand{\poly}{{\rm poly}}

\newcommand{\remove}[1]{}

\usepackage{times}
\usepackage[varg]{txfonts} 
\renewcommand{\mathbb}{\varmathbb}

\usepackage{microtype}

\newcommand{\Esymb}{\mathbb{E}}
\newcommand{\Psymb}{\mathbb{P}}

\DeclareMathOperator*{\E}{\Esymb}

\DeclareMathOperator*{\ProbOp}{\Psymb r}
\renewcommand{\Pr}{\ProbOp}

\newcommand{\mper}{\,.}
\newcommand{\mcom}{\,,}

\newcommand{\R}{\mathbb{R}}

\newcommand{\cI}{{\cal I}}

\newcommand{\Set}[1]{\left\{#1\right\}}

\newcommand{\bits}{\left\{0,1\right\}}
\newcommand{\defeq}{\stackrel{\small \mathrm{def}}{=}}

\newcommand{\opt}{\mathrm{opt}}
\newcommand{\rd}{\mathrm{d}}
\newcommand{\tv}{\mathrm{tv}}

\newcommand{\Sfp}{Self-fulfilling prophecy}
\newcommand{\vendor}{vendor}

\newcommand{\Scomp}{S^c}

\newcommand{\network}{data owner}

\newcommand\bias{\mathrm{bias}}
\newcommand\EM{\mathrm{EM}}
\newcommand\ExM{\mathrm{E}}

\newcommand\dtoD{D,d}
\newcommand\dtoDtv{D_\tv,\, d}
\newcommand\dtoDinfty{D_\infty,\, d}
\newcommand\ad{\mathrm{ad}}

\begin{document}

\title{Fairness Through Awareness}
\vspace{5mm}

\author{Cynthia Dwork\thanks{Microsoft Research Silicon Valley, Mountain View,
CA, USA. Email: {\tt dwork@microsoft.com}}
\and Moritz Hardt\thanks{
IBM Research Almaden, San Jose, CA, USA. Email: {\tt mhardt@us.ibm.com}. Part of
this work has been done while the author visited Microsoft Research Silicon
Valley.}
\and Toniann Pitassi\thanks{University of Toronto, Department of Computer Science.
Supported by NSERC. Email: {\tt toni@cs.toronto.edu}. Part of this work has been done while the author visited
Microsoft Research Silicon Valley.}
\and
Omer Reingold\thanks{Microsoft Research Silicon Valley, Mountain View,
CA, USA. Email: {\tt omer.reingold@microsoft.com}}
\and Richard Zemel\thanks{University of Toronto, Department of Computer Science.
Supported by NSERC.
Email: {\tt zemel@cs.toronto.edu}. Part of this work has been done while the author visited
Microsoft Research Silicon Valley.}}

\maketitle
\thispagestyle{empty}

\begin{abstract}
We study {\em fairness in classification}, where individuals are classified, 
e.g., admitted to a university, and the goal is to prevent
discrimination against individuals based on their membership in
some group, while maintaining utility for the classifier (the
university).  
The main conceptual contribution of this paper
    is a framework for fair classification comprising (1) a
    (hypothetical) task-specific metric for determining the degree to
    which individuals are similar with respect to the classification
    task at hand; (2) an algorithm for maximizing utility subject to
    the {\em fairness constraint}, that similar individuals are treated
    similarly.
We also present an adaptation of our approach to achieve the
complementary goal of ``fair affirmative action,'' which
    guarantees {\em statistical parity} (i.e., the demographics
of the set of individuals receiving any classification are the same as
the demographics of the underlying population), while treating similar
    individuals as similarly as possible.  
Finally, we discuss the relationship of fairness to privacy: when
fairness implies privacy, and how tools developed in
the context of differential privacy may be applied to fairness.
\end{abstract}

\vfill
\pagebreak

\setcounter{page}{1}

\section{Introduction}
\newcommand{\ep}{Embarrassing Product}

In this work, we study fairness in classification.
Nearly all classification tasks face the challenge of achieving utility in
classification for some purpose, while at the same time preventing
discrimination against protected population subgroups.  A motivating example
is membership in a racial minority in the context of banking.  An article in
The Wall Street Journal (8/4/2010) describes the practices of a credit card company
and its use of a tracking network 
to learn detailed demographic
information about each visitor to the site, such as approximate income, where
she shops, the fact that she rents children's videos, and so on. According to
the article, this information is used to ``decide which credit cards to show
first-time visitors'' to the web site, raising the concern of {\em steering,}
namely the (illegal) practice of guiding members of minority groups into less
advantageous credit offerings~\cite{WSJ:8-4-10}.

We provide a normative approach to fairness in classification and a
framework for achieving it.  Our framework permits us to formulate the
question as an optimization problem that can be solved by a linear program.
In keeping with the motivation of fairness in online advertising, our approach
will permit the entity that needs to classify individuals, which we call the
{\it \vendor{}}, as much freedom as possible, without knowledge of or trust in
this party.  This allows the \vendor{} to benefit from investment in data
mining and market research in designing its classifier, while our absolute
guarantee of fairness frees the \vendor{} from regulatory concerns.

Our approach is centered around the notion of a task-specific {\em
  similarity metric} describing the extent to which pairs of
individuals should be regarded as similar for the classification task
at hand.\footnote{%
Strictly speaking, we only require a function $d\colon
  V\times V\to \mathbb{R}$ where $V$ is the set of individuals,
  $d(x,y)\ge 0$, $d(x,y)=d(y,x)$ and $d(x,x)=0.$}
The similarity metric expresses ground truth.
When ground truth
is unavailable, the metric may reflect the ``best'' available
approximation as agreed upon by society.  Following established
tradition~\cite{Rawls01}, the metric is assumed to be public and open
to discussion and continual refinement.  Indeed, we envision that,
typically, the distance metric would be externally imposed, for
example, by a regulatory body, or externally proposed, by a civil
rights organization.

The choice of a metric need not determine (or even suggest) a
particular classification scheme. There can be many classifiers consistent
with a single metric. Which classification scheme is chosen in the end is
a matter of the \vendor 's utility function which we take into account. To give
a concrete example, consider a metric that expresses which individuals have
similar credit worthiness. One advertiser may wish to target a specific
product to individuals with low credit, while another advertiser may seek
individuals with good credit.

\subsection{Key Elements of Our Framework}

\paragraph{Treating similar individuals similarly.}
We capture fairness by the principle that any two individuals who are similar
{\em with respect to a particular task} should be classified similarly.  In
order to accomplish this \emph{individual-based} fairness, we assume a
distance metric that defines the similarity between the individuals. This is
the source of ``awareness" in the title of this paper. 
We formalize this guiding principle as a \emph{Lipschitz} condition on 
the classifier. 
In our approach a
classifier is a randomized mapping from individuals to outcomes, or
equivalently, a mapping from individuals to distributions over outcomes. 
The Lipschitz condition requires that any two individuals $x,y$ that
are at distance $d(x,y)\in[0,1]$ map to distributions $M(x)$ and $M(y),$
respectively, such that the statistical distance between $M(x)$ and $M(y)$ is
at most $d(x,y).$ In other words, the distributions over outcomes observed by $x$ and $y$ 
are indistinguishable up to their distance $d(x,y).$ 
\paragraph{Formulation as an optimization problem.}
We consider the natural optimization problem of constructing fair (i.e.,
Lipschitz) classifiers that minimize the expected utility loss of
the \vendor{}.
We observe that this optimization problem can be expressed as a linear program
and hence solved efficiently. Moreover, this linear program and its dual
interpretation will be used heavily throughout our work.

\paragraph{Connection between individual fairness and group fairness.}
\emph{Statistical parity} is the property that the demographics of those
receiving positive (or negative) classifications are identical to the
demographics of the population as a whole. Statistical parity speaks to
\emph{group fairness} rather than individual fairness, and appears 
desirable, as it equalizes outcomes across protected and non-protected
groups.
However, we demonstrate its inadequacy as
a notion of fairness through several examples in which statistical parity is
maintained, but from the point of view of an individual, the outcome is
blatantly unfair.
While statistical parity (or group fairness) is insufficient by itself, we
investigate conditions under which  our notion of fairness implies
statistical parity.  In Section~\ref{sec:lip-sp}, we give conditions on the
similarity metric, via an Earthmover distance, such that fairness for
individuals (the Lipschitz condition) yields group fairness (statistical
parity). More precisely, we show that the Lipschitz condition implies
statistical parity between two groups if and only if the Earthmover distance
between the two groups is small. This characterization is an important tool in
understanding the consequences of imposing the Lipschitz condition.

\paragraph{Fair affirmative action.}
In Section~\ref{sec:preferential}, we give
techniques for forcing statistical parity when it is not implied by the
Lipschitz condition (the case of preferential treatment), while preserving as
much fairness for individuals as possible. We interpret these results as
providing a way of achieving \emph{fair affirmative action}.

\paragraph{A close relationship to privacy.}
We observe that our definition of fairness is a generalization of the notion
of differential privacy~\cite{Dwork06,DworkMcNiSm06}.  We draw an analogy
between individuals in the setting of fairness and databases in the setting of
differential privacy. In Section~\ref{sec:expmechanism} we build on this
analogy and exploit techniques from differential privacy to develop a more
efficient variation of our fairness mechanism. We prove that our solution has
small error when the metric space of individuals has small doubling dimension,
a natural condition arising in machine learning applications. We also prove a
lower bound showing that \emph{any} mapping satisfying the Lipschitz condition
has error that scales with the doubling dimension. Interestingly, these results 
also demonstrate a quantiative trade-off between fairness and utility.
Finally, we touch on the extent to which fairness can 
hide information from the advertiser 
in the context of online advertising.

\paragraph{Prevention of certain evils.}
We remark that our notion of fairness interdicts a catalogue of discriminatory
practices including the following, described
 in Appendix~\ref{sec:taxonomy}: redlining; reverse redlining;
discrimination based on redundant encodings of membership in the protected
set; cutting off business with a segment of the population in which membership
in the protected set is disproportionately high; doing business with the
``wrong'' subset of the protected set (possibly in order to prove a point);
and ``reverse tokenism.''

\subsection{Discussion: The Metric}
As noted above, the metric should (ideally) capture ground truth.
Justifying the {\em availability of} or {\em access to} 
the distance metric in various settings
is one of the most challenging aspects of our framework,
and in reality the metric used will most likely only be society's 
current best approximation to the truth.
Of course, metrics are employed, implicitly or explicitly, in
many classification settings, such as college admissions procedures,
advertising (``people who buy X and live in zipcode Y are similar to
people who live in zipcode Z and buy W''), and loan applications
(credit scores).  Our work advocates for making these metrics public.

An intriguing example of an existing metric designed for the health
care setting is part of the AALIM project~\cite{AALIM},
whose goal is
to provide a decision support system for cardiology that
helps a physician in finding a suitable diagnosis for a patient based
on the consensus opinions of other physicians who have looked at
similar patients in the past. Thus the system requires an accurate
understanding of which patients are \emph{similar} based on information from
multiple domains such as cardiac echo videos, heart sounds, ECGs and
physicians' reports.  AALIM seeks to ensure that individuals with similar
health characteristics receive similar treatments from physicians. This work
could serve as a starting point in the fairness setting, although it does not
(yet?) provide the distance metric that our approach requires.  We discuss
this further in Section~\ref{sec:metric}.

Finally, we can envision classification situations in which
it is desirable to ``adjust'' or
otherwise ``make up'' a metric, and use this synthesized
metric as a basis for determining which pairs of individuals should
be classified similarly.\footnote{%
This is consistent with the practice, in some college admissions offices,
of adding a certain number of points to SAT scores of students in 
disadvantaged groups.}  Our machinery is agnostic as to the ``correctness''
of the metric, and so can be employed in these settings as well.

\subsection{Related Work}
There is a broad literature on fairness, notably in social choice theory,
game theory, economics, and law.
Among the most relevant are theories of fairness and algorithmic
approaches to apportionment; see, for example, the following books: 
H. Peyton Young's {\em Equity}, John Roemer's {\em Equality of Opportunity}
and {\em Theories of Distributive Justice}, as well as John Rawls' 
{\em A Theory of Justice} and {\em Justice as Fairness: A Restatement}.
Calsamiglia~\cite{Calsamiglia05} explains, 
\begin{quote}
``Equality of opportunity defines an important welfare
criterion in political philosophy and policy analysis.  Philosophers define
equality of opportunity as the requirement that an individual's well being be
independent of his or her irrelevant characteristics.  The difference
among philosophers is mainly about which characteristics should be
considered irrelevant.  Policymakers, however, are often called upon to 
address more specific questions: How should admissions policies be designed 
so as to provide equal opportunities for college? Or how should tax
schemes be designed so as to equalize opportunities for income?  
These are called local distributive justice problems, because each 
policymaker is in charge of achieving equality of opportunity to a 
specific issue.''
\end{quote}
In general, local solutions do not, taken together, solve the global problem:
``There is no mechanism comparable to the invisible hand of the market for
coordinating distributive justice at the micro into just outcomes at the macro
level''~\cite{Young95}, (although Calsamiglia's work treats exactly this
problem~\cite{Calsamiglia05}).  Nonetheless, our work is decidedly ``local,''
both in the aforementioned sense and in our definition of fairness.  To our
knowledge, our approach differs from much of the literature in our fundamental
skepticism regarding the \vendor ; we address this by separating the \vendor{}
from the \network , leaving classification to the latter.

Concerns for ``fairness'' also arise in many contexts in computer science,
game theory, and economics.  For example, in the distributed computing
literature, one meaning of fairness is that a process that attempts infinitely
often to make progress eventually makes progress. One quantitative meaning of
{\em unfairness} in scheduling theory is the maximum, taken over all members
of a set of long-lived processes, of the difference between the actual load on
the process and the so-called {\em desired} load (the desired load is a
function of the tasks in which the process participates)~\cite{AjtaiANRSW98};
other notions of fairness appear in \cite{BansalS06,Feige08,FeigeT11}, to name
a few.  For an example of work incorporating fairness into game theory and
economics see the eponymous paper~\cite{Rabin93}.

\section{Formulation of the Problem}
\label{sec:formulation}

In this section we describe our setup in its most basic form. We shall later
see generalizations of this basic formulation.
{\em Individuals} are the objects to be classified; we denote the set of
individuals by $V$.
In this paper we consider classifiers that map individuals
to outcomes. We denote the set of outcomes by~$A.$ In the simplest
non-trivial case~$A=\bits.$ To ensure fairness, we will consider randomized
classifiers mapping individuals to distributions over outcomes.
To introduce our notion of fairness we assume the existence of a
metric on individuals $d\colon V\times V\to\mathbb{R}.$
We will consider randomized mappings $M\colon V\to \Delta(A)$
from individuals to probability distributions
over outcomes. Such a mapping naturally describes a randomized classification
procedure: to classify $x \in V$ choose an outcome
$a$ according to the distribution $M(x)$.
We interpret the goal of ``mapping similar people similarly''
to mean that the distributions assigned to similar people are similar.
Later we will discuss two specific measures of similarity of 
distributions, $D_\infty$ and $D_\tv$,
of interest in this work.
\begin{definition}[Lipschitz mapping]
\label{def:Lipschitz-map}
A mapping $M\colon V\to\Delta(A)$ satisfies the \emph{$(\dtoD)$-Lipschitz}
property if for every $x,y\in V,$ we have
\begin{equation}\label{eq:lipschitz}
D(Mx,My)\le d(x,y)\mper
\end{equation}
When $D$ and $d$ are clear from the context we will refer to this
simply as the \emph{Lipschitz} property. 
\end{definition}
\remove{
\begin{definition}\label{def:Lipschitz}
Let $D$ be a distance measure between distributions.
A randomized mapping $F\colon V\to A$ satisfies the \emph{$d$-Lipschitz}
property if for every two individuals $x,y\in V$ and every set of outcomes
$O\subseteq A,$
\[
\frac{\Pr\{F(x)\in O\}}{\Pr\{F(y)\in O\}}
\le \exp(d(x,y))\mper
\]
\end{definition}
}
\remove{
We think of two individuals $x,y$ as similar if $d(x,y)\ll 1.$ In this case,
the above condition ensures that $x$ and $y$ map to similar distributions
over~$A.$ On the other hand, when $x,y$ are very dissimilar, i.e., $d(x,y)\gg
1,$ the condition imposes only a weak constraint on the two corresponding
distributions over outcomes.
}
We note that there always exists a Lipschitz classifier, for example, by
mapping all individuals to the same distribution over~$A.$ Which classifier we
shall choose thus depends on a notion of utility. We capture utility using a
\emph{loss function}~$L\colon V\times A\to\R.$ This setup naturally leads
to the optimization problem:
\begin{quote}
Find a mapping from individuals to distributions over outcomes
that minimizes expected loss subject to the Lipschitz condition.
\end{quote}

\subsection{Achieving Fairness}
\label{sec:achieving}
 Our fairness definition leads to an optimization problem in which we minimize
an arbitrary loss function $L\colon V\times A\to\mathbb{R}$ while achieving
the~$(\dtoD)$-Lipschitz property for a given metric~$d\colon V\times V\to\R$.
We denote by $\cI$ an instance of our problem consisting
of a metric~$d\colon V\times V\to\R,$ and a loss function $L\colon V\times
A\to\R.$
We denote the optimal value of the minimization problem by $\mathrm{opt}(\cI)$,
as formally defined in Figure~\ref{fig:opt}. We will also write the mapping $M\colon V\to\Delta(A)$ as
$M=\{\mu_x\}_{x\in V}$ where $\mu_x=M(x)\in\Delta(A).$

\begin{figure}[h]
\begin{align} \label{eq:fairness-lp}
\mathrm{opt}(\cI)\defeq\quad\min_{\{\mu_x\}_{x\in V}}\quad
& \E_{x\sim V}\E_{a\sim\mu_x} L(x,a)  \\
\text{subject to}\quad
&\forall x,y\in V,\colon\quad
D(\mu_x,\mu_y) \le d(x,y)
\label{eq:lip-c}\\
& \forall x\in V\colon\quad\mu_x\in \Delta(A)
\end{align}
\caption{The Fairness LP: Loss minimization subject to fairness constraint}
\label{fig:opt}
\end{figure}

\paragraph{Probability Metrics}
The first choice for $D$ that may come to mind is the statistical distance:
Let $P,Q$ denote probability measures on a finite domain~$A.$ The
\emph{statistical distance} or \emph{total variation norm} between $P$ and $Q$
is denoted by
\begin{equation}\label{eq:tvnorm}
D_\tv(P,Q)=\frac12\sum_{a\in A}|P(a)-Q(a)|\mper
\end{equation}

The following lemma is easily derived from the definitions of $\opt(\cI)$ and $D_\tv$.
\begin{lemma}\label{lem:opt}
Let $D = D_\tv$.
Given an instance $\cI$ we can compute $\opt(\cI)$
with a linear program of size $\poly(|V|,|A|).$
\end{lemma}

\remove{
We also introduce a distance measure sometimes called
\emph{relative $\ell_\infty$ metric}:
\begin{equation}\label{eq:relinfty}
D_\infty(P,Q)
=\sup_{a\in A}\log\left(
\max\left\{
\frac{P(a)}{Q(a)},\frac{Q(a)}{P(a)}\right\}
\right)
\mper
\end{equation}

\begin{definition}[Lipschitz mapping]
\label{def:Lipschitz-map}
A mapping $M\colon V\to\Delta(A)$ satisfies the \emph{$(\dtoD)$-Lipschitz}
property if for every $x,y\in V,$ we have
\begin{equation}\label{eq:lipschitz}
D(Mx,My)\le d(x,y)\mper
\end{equation}
When $D$ and $d$ are clear from the context we will refer to this simply as
the \emph{Lipschitz} property. In this case we may call $M$ a
\emph{Lipschitz mapping}.
\end{definition}
\begin{remark}
We note that Definition~\autoref{def:Lipschitz} is just a special case of the
previous definition by taking $D=D_\infty$ and thinking of a randomized
classifier as a mapping $M\colon V\to\Delta(A)$ in the natural way. The choice
of $D_\infty$ in Definition~\ref{def:Lipschitz} is not crucial for its use as
a fairness definition. We could have also chosen $D_\tv$ instead. In fact, we
will later also study $(D_\tv,d)$-Lipschitz mappings. It is however important
to note that the choice of a probability metric affects what the right
scaling of the metric~$d$ is.
\end{remark}
}

\begin{remark}
When dealing with the set~$V$, we have assumed that $V$ is
the set of real individuals (rather than the potentially huge set of
all possible encodings of individuals).
More generally, we may only have
access to a subsample from the set of interest.
In such a case, there is the additional
challenge of extrapolating a classifier over the entire set.
\end{remark}

A weakness of using $D_\tv$ as the distance measure on distributions, it that we should then assume that the distance
metric (measuring distance between individuals) is scaled such that for similar individuals $d(x,y)$ is very close to zero,
while for very dissimilar individuals $d(x,y)$ is close to one. A potentially better choice for $D$ in this respect is
sometimes called
\emph{relative $\ell_\infty$ metric}:
\begin{equation}\label{eq:relinfty}
D_\infty(P,Q)
=\sup_{a\in A}\log\left(
\max\left\{
\frac{P(a)}{Q(a)},\frac{Q(a)}{P(a)}\right\}
\right)
\mper
\end{equation}
With this choice
we think of two individuals $x,y$ as similar if $d(x,y)\ll 1.$ In this case,
the Lipschitz condition in Equation~\ref{eq:lipschitz}
ensures that $x$ and $y$ map to similar distributions
over~$A.$ On the other hand, when $x,y$ are very dissimilar, i.e., $d(x,y)\gg
1,$ the condition imposes only a weak constraint on the two corresponding
distributions over outcomes.
\begin{lemma}\label{lem:optinfty}
Let $D = D_\infty$.
Given an instance $\cI$ we can compute $\opt(\cI)$
with a linear program of size $\poly(|V|,|A|).$
\end{lemma}

\begin{proof}
We note that the objective function and the
first constraint are indeed linear in the
variables~$\mu_x(a),$
as the first constraint boils down to requirements of the form
$\mu_x(a) \le e^{d(x,y)}\mu_y(a)$.
The second constraint $\mu_x\in\Delta(A)$ can easily be
rewritten as a set of linear constraints.
\end{proof}

\paragraph{Notation.}
Recall that we often write the mapping $M\colon V\to\Delta(A)$ as
$M=\{\mu_x\}_{x\in V}$ where $\mu_x=M(x)\in\Delta(A).$
In this case, when $S$ is a distribution over $V$ we denote by $\mu_S$ the
distribution over~$A$ defined as
$\mu_S(a) = \E_{x\sim S}\mu_x(a)$ where $a\in A\mper$

\paragraph{Useful Facts}
It is not hard to check that both~$D_\tv$ and~$D_\infty$ are metrics with the
following properties.
\begin{lemma}\label{lem:tv2inf}
$D_\tv(P,Q)\le 1-\exp(-D_\infty(P,Q))\le  D_\infty(P,Q)$
\end{lemma}
\begin{fact}\label{fac:tvconv}
For any three distributions $P,Q,R$ and non-negative numbers
$\alpha,\beta\ge0$ such that $\alpha+\beta=1,$ we have $D_\tv(\alpha P+\beta
Q,R)\le \alpha D_\tv(P,R)+\beta D_\tv(Q,R).$
\end{fact}

\paragraph{Post-Processing.}
An important feature of our definition is that it behaves well with respect to
\emph{post-processing}. Specifically, if $M\colon V\to\Delta(A)$ is
$(D,d)$-Lipschitz for $D\in\{D_\tv,D_\infty\}$ and $f\colon A\to B$ is any
possibly randomized function from $A$ to another set $B,$ then the composition
$f\circ M\colon V\to\Delta(B)$ is a $(D,d)$-Lipschitz mapping. This would in particular be useful in the setting of
the example in Section~\ref{sec:adexample}.

\subsection{Example: Ad network}
\label{sec:adexample}

Here we expand on the example of an advertising network mentioned in
the Introduction.  We explain how the Fairness LP provides a fair
solution protecting against the evils described in
Appendix~\ref{sec:taxonomy}.  The Wall Street Journal
article~\cite{WSJ:8-4-10} describes how the [x+1] tracking network
collects demographic information about individuals, such as their
browsing history, geographical location, and shopping behavior, and
utilizes this to assign a person to one of 66 groups.  For example,
one of these groups is ``White Picket Fences,'' a market segment with
median household income of just over \$50,000, aged 25 to 44 with
kids, with some college education, etc.  Based on this assignment to a
group, CapitalOne decides which credit card, with particular terms of
credit, to show the individual.  In general we view a classification
task as involving two distinct parties: the {\it \network} is a
trusted party holding the data of individuals, and the {\it vendor} is
the party that wishes to classify individuals.  The loss function may
be defined solely by either party or by both parties in collaboration.
In this example, the data owner is the ad network [x+1], and the
vendor is CapitalOne.

The ad network ([x+1]) maintains a mapping from
individuals into categories. We can think of these categories as outcomes, as
they determine which ads will be shown to an individual. In order to comply
with our fairness requirement, the mapping from individuals into categories
(or outcomes) will have to be randomized and satisfy the Lipschitz property
introduced above. Subject to the Lipschitz constraint, the \vendor{} can
still express its own belief as to how individuals should be assigned to
categories using the loss function. However, since the Lipschitz condition is
a hard constraint there is no possibility of discriminating between
individuals that are deemed similar by the metric.
In particular, this will disallow arbitrary distinctions between
protected individuals, thus preventing both reverse tokenism
and the self-fulfilling prophecy (see Appendix~\ref{sec:taxonomy}). 
In addition, the metric can eliminate the
existence of redundant encodings of certain attributes thus also preventing
redlining of those attributes. In Section~\ref{sec:lip-sp} we will see a
characterization of which attributes are protected by the metric in this way.

\subsection{Connection to Differential Privacy}
\label{sec:privacy}
Our notion of fairness may be viewed as a generalization of
differential
privacy~\cite{Dwork06,DworkMcNiSm06}.  As it turns
out our notion can be seen as a generalization of differential privacy. To see
this, consider a simple setting of differential privacy  where a
\emph{database curator} maintains a database~$x$  (thought of as a subset of
some universe $U$) and a data analyst is allowed to ask a query $F\colon V\to
A$ on the database.  Here we denote the set of databases by $V=2^U$ and the
range of the query by~$A.$ A mapping $M\colon V\to\Delta(A)$
satisfies \emph{$\epsilon$-differential privacy} if and only if $M$ satisfies
the $(\dtoDinfty)$-Lipschitz property, where,
letting
$x\triangle y$ denote
the symmetric difference between $x$ and $y$,
we define
$d(x,y)\defeq \epsilon |x\triangle y|$.

The utility loss of the analyst for getting an answer $a\in A$ from the
mechanism is defined as $L(x,a)=d_A(Fx,a),$ that is distance of the true
answer from the given answer. Here distance refers to some distance measure in
$A$ that we described using the notation~$d_A.$ For example, when
$A=\mathbb{R}$, this could simply be~$d_A(a,b)=|a-b|.$
The optimization problem~(\ref{eq:fairness-lp}) in Figure~\ref{fig:opt}
({\it i.e.},
$\mathrm{opt}(\cI)\defeq \min
\E_{x\sim V}\E_{a\sim\mu_x} L(x,a)$)
now defines the optimal
differentially private mechanism in this setting. We can draw a conceptual
analogy between the utility model in differential privacy and that in
fairness. If we think of outcomes as representing information about an
individual, then the vendor wishes to receive what she believes is the most
``accurate'' representation of an individual. This is quite similar to the
goal of the analyst in differential privacy.

In the current work we
deal with more general metric
spaces than in differential privacy. Nevertheless, we later see (specifically in
Section~\ref{sec:expmechanism}) that some of the techniques
used in differential privacy carry over to the fairness setting.

\section{Relationship between Lipschitz property and statistical parity}
\label{sec:lip-sp}
In this section we discuss the relationship between the Lipschitz property
articulated in Definition~\ref{def:Lipschitz-map} and \emph{statistical
parity}. As we discussed earlier, statistical parity is insufficient as a
general notion of fairness. Nevertheless statistical parity can have several
desirable features,
{\it e.g.}, as described in Proposition~\ref{nice} below.  In this section we
demonstrate that the Lipschitz condition naturally implies statistical parity
between certain subsets of the population.

Formally, statistical parity is the following property.

\begin{definition}[Statistical parity]
We say that a mapping $M\colon V\to\Delta(A)$ satisfies \emph{statistical
parity} between distributions $S$ and $T$ up to bias $\epsilon$ if
\begin{equation}
D_\tv(\mu_S,\mu_T)\le\epsilon\mper
\end{equation}
\end{definition}


\begin{proposition}
\label{nice}
Let $M\colon V\to\Delta(A)$ be a mapping that satisfies statistical
parity between two sets~$S$ and $T$ up to bias~$\epsilon.$
Then, for every set of outcomes~$O\subseteq A,$
we have the following two properties.
\begin{enumerate}
\item $\left|\Pr\Set{M(x)\in O \mid x\in S}
- \Pr\Set{M(x)\in O\mid x\in T}\right|\le\epsilon,$
\item $\left|\Pr\Set{x\in S \mid M(x)\in O}-
 \Pr\Set{x\in T\mid M(x)\in O})\right|\le\epsilon\mper$
\end{enumerate}
\end{proposition}
Intuitively, this proposition says that if $M$ satisfies statistical parity,
then members of $S$ are equally likely to observe a set of outcomes as are
members of~$T.$ Furthermore, the fact that an individual observed a particular
outcome provides no information as to whether the individual is a member
of~$S$ or a member of~$T.$ We can always choose $T=S^c$ in which case we
compare $S$ to the general population.

\subsection{Why is statistical parity insufficient?}
Although in some cases statistical parity appears to be desirable -- in
particular, it neutralizes redundant encodings --  we 
now argue its inadequacy as a notion of fairness,
presenting three examples in which statistical parity is
maintained, but from the point of view of an individual, the outcome
is blatantly unfair.  In describing these examples, we let $S$ denote
the protected set and $\Scomp$ its complement.

\begin{description}
\item[{\em Example 1: Reduced Utility.}]\label{ex:poor_utility}
  Consider the following scenario. 
  Suppose in the culture of $S$ the most talented students are
  steered toward science and engineering and the less talented
  are steered toward finance, while in the culture of $\Scomp$
  the situation is reversed: the most talented are steered toward finance
  and those with less talent are steered toward engineering.
  An organization ignorant of the culture of $S$ and seeking the most talented people may select for
  ``economics,'' arguably choosing the wrong subset of $S$, even while
  maintaining parity.
  Note that this poor outcome can occur in a ``fairness
  through blindness'' approach -- the errors come from {\em ignoring} membership
  in $S$.
\item[{\em Example 2: \Sfp .}]  This is when unqualified members of $S$ are
  chosen, in order to ``justify'' future discrimination against~$S$ (building a
  case that there is no point in ``wasting'' resources on $S$).
  Although senseless, this is an example of something pernicious that is not
  ruled out by statistical parity, showing the weakness of this notion.
A variant of this apparently occurs in selecting candidates for interviews: the hiring practices of certain firms are audited to ensure sufficiently many {\em interviews} of minority candidates, but less care is taken to ensure that the best minorities -- those that might actually compete well with the better non-minority candidates -- are invited~\cite{Zarsky11}.
\item[{\em Example 3: Subset Targeting}.] Statistical parity for $S$ does not imply
  statistical parity for subsets of $S$. This can be maliciously exploited in
  many ways. For example, consider an advertisement for a product $X$ which is
  targeted to members of $S$ that are likely to be interested in $X$ and to
  members of $\Scomp$ that are very unlikely to be interested in $X$. {\em Clicking}
  on such an ad may be strongly correlated with membership in $S$ (even if
  exposure to the ad obeys statistical parity).
\end{description}

\subsection{Earthmover distance: Lipschitz versus statistical parity}

A fundamental question that arises in our approach is:
\emph{When does the Lipschitz
condition imply statistical parity between two
distributions~$S$ and~$T$ on $V$?}
We will see that the answer to this question is closely  related to the
\emph{Earthmover distance} between $S$ and $T$, which we will define shortly.

The next definition formally introduces the quantity that we will
study, that is, the extent to which any Lipschitz mapping can
violate statistical parity.
In other words, we answer the question,
``How biased with respect to $S$ and $T$ might the solution of the
fairness LP be, in the worst case?''
\begin{definition}[Bias]
We define
\begin{equation}
\bias_{\dtoD}(S,T)\defeq\max \mu_S(0)-\mu_T(0)\mcom
\end{equation}
where the maximum is taken over all $(\dtoD)$-Lipschitz mappings
$M=\{\mu_x\}_{x\in V}$ mapping~$V$ into $\Delta(\{0,1\}).$
\end{definition}

Note that $\bias_{\dtoD}(S,T)\in[0,1].$ Even though in the definition
we restricted ourselves to mappings into distributions over~$\{0,1\},$
it turns out that this is without loss of generality, as we show next.

\begin{lemma}
\label{lem:zo}
Let $D\in\{D_\tv,D_\infty\}$ and let
$M\colon V\to\Delta(A)$ be any $(\dtoD)$-Lipschitz mapping. Then,
$M$ satisfies statistical parity between $S$ and $T$ up to $\bias_{\dtoD}(S,T).$
\end{lemma}
\begin{proof}
Let $M=\{\mu_x\}_{x\in V}$ be any $(\dtoD)$-Lipschitz mapping into~$A.$
We will construct a $(\dtoD)$-Lipschitz mapping $M'\colon
V\to\Delta(\{0,1\})$ which has the same bias between $S$ and $T$ as $M.$

Indeed, let $A_S=\{a\in A\colon \mu_S(a)>\mu_T(a)\}$ and let $A_T=A_S^c.$ Put
$\mu'_x(0)=\mu_x(A_S)$ and $\mu'_x(1)=\mu_x(A_T).$
We claim that $M'=\{\mu_x'\}_{x\in V}$ is a $(\dtoD)$-Lipschitz mapping.
In both cases $D\in\{D_\tv,D_\infty\}$ this follows directly from the
definition.
On the other hand, it is easy to see that
\[
D_\tv(\mu_S,\mu_T)=D_\tv(\mu'_S,\mu'_{T})=\mu'_S(0)-\mu'_T(0)\le
\bias_{\dtoD}(S,T)\mper
\qedhere
\]
\end{proof}

\paragraph{Earthmover Distance.}
We will presently relate~$\bias_{\dtoD}(S,T)$ for $D\in\{D_\tv,D_\infty\}$
to certain \emph{Earthmover distances} between $S$ and $T$, which
we define next.

\begin{definition}[Earthmover distance]
Let $\sigma\colon V\times V\to\mathbb{R}$ be a nonnegative distance function.
The $\sigma$-Earthmover distance between two distributions
$S$ and $T$, denoted
$\sigma_{\EM}(S,T),$ is
defined as the value of the so-called \emph{Earthmover LP}:
\begin{align*}
\sigma_\EM(S,T)\defeq\quad\min\quad & \sum_{x,y\in V} h(x,y)\sigma(x,y)\\
\text{subject to}\quad
& \sum_{y\in V} h(x,y) = S(x)\\
& \sum_{y\in V} h(y,x) = T(x)\\
& h(x,y)\ge 0
\end{align*}
\end{definition}

We will need the following standard lemma, which simplifies the definition of
the Earthmover distance in the case where $\sigma$ is a metric.
\begin{lemma}\label{lem:em-metric}
Let $d\colon V\times V\to\R$ be a metric. Then,
\begin{align*}
d_\EM(S,T)=\quad\min\quad & \sum_{x,y\in V} h(x,y)d(x,y)\\
\mathrm{subject\, to}\quad
& \sum_{y\in V} h(x,y) = \sum_{y\in V}
h(y,x)+S(x)-T(x)\\
& h(x,y)\ge 0
\end{align*}
\end{lemma}
\begin{theorem}\label{thm:em-tv}
Let $d$ be a metric. Then,
\begin{equation}\label{eq:tv2em}
\bias_{\dtoDtv}(S,T)\le d_{\EM}(S,T)\mper
\end{equation}
If furthermore $d(x,y)\le1$ for all $x,y,$ then we have
\begin{equation}\label{eq:em2tv}
\bias_{\dtoDtv}(S,T)\ge d_{\EM}(S,T)\mper
\end{equation}
\end{theorem}
\begin{proof}
The proof is by linear programming duality. We can express $\bias_{\dtoDtv}(S,T)$ as
the following linear program:
\begin{align*}
\bias(S,T) =\quad
\max  \quad& \sum_{x\in V}S(x)\mu_x(0)- \sum_{x\in V}T(x)\mu_x(0)  \\
\text{subject to}\quad
& \mu_x(0) -\mu_y(0) \le d(x,y) \\
& \mu_x(0) + \mu_x(1) = 1 \\
& \mu_x(a)\ge0
\end{align*}
Here, we used the fact that
\begin{equation*}
D_\tv(\mu_x,\mu_y)\le d(x,y)\quad\Longleftrightarrow\quad
\left|\mu_x(0)-\mu_y(0)\right|\le d(x,y)\mper
\end{equation*}
The constraint on the RHS is enforced in the linear program above by the two
constraints $\mu_x(0)-\mu_y(0)\le d(x,y)$ and $\mu_y(0)-\mu_x(0)\le d(x,y).$

We can now prove~(\ref{eq:tv2em}). Since $d$ is a metric, we can apply
Lemma~\ref{lem:em-metric}. Let $\{f(x,y)\}_{x,y\in V}$ be a solution to the LP
defined in Lemma~\ref{lem:em-metric}. By putting $\epsilon_x=0$ for all $x\in
V,$ we can extend this to a feasible solution to the LP defining $\bias(S,T)$
achieving the same objective value. Hence, we have $\bias(S,T)\le d_\EM(S,T).$

Let us now prove~(\ref{eq:em2tv}), using the assumption that $d(x,y)\le1.$
To do so, consider dropping the constraint that $\mu_x(0)+\mu_x(1)=1$ and
denote by $\beta(S,T)$ the resulting LP:
\begin{align*}
\beta(S,T) \defeq\quad\max\quad
& \sum_{x\in V}S(x)\mu_x(0)- \sum_{x\in V}T(x)\mu_x(0)  \\
\text{subject to}\quad & \mu_x(0) -\mu_y(0) \le d(x,y) \\
& \mu_x(0)\ge0
\end{align*}
It is clear that $\beta(S,T)\ge\bias(S,T)$ and we claim that in
fact $\bias(S,T)\ge\beta(S,T).$ To see this, consider any solution
$\{\mu_x(0)\}_{x\in V}$ to $\beta(S,T).$  Without changing the objective value
we may assume that $\min_{x\in V}\mu_x(0)=0.$ By our assumption that
$d(x,y)\le1$ this means that $\max_{x\in V}\mu_x(0)\le1.$ Now put
$\mu_x(1)=1-\mu_x(0)\in[0,1].$ This gives a solution to
$\bias(S,T)$ achieving the same objective value. We therefore have,
\[
\bias(S,T)=\beta(S,T)\mper
\]
On the other hand, by strong LP duality, we have
\begin{align*}
\beta(S,T)=\quad\min\quad & \sum_{x,y\in V} h(x,y)d(x,y)\\
\text{subject to}\quad
&  \sum_{y\in V} h(x,y) \ge \sum_{y\in V} h(y,x)+S(x)-T(x)\\
& h(x,y)\ge 0
\end{align*}
It is clear that in the first constraint we must have equality in any optimal
solution. Otherwise we can improve the objective value by decreasing some
variable $h(x,y)$ without violating any constraints.

Since $d$ is a metric we can now apply Lemma~\ref{lem:em-metric} to conclude
that $\beta(S,T)=d_\EM(S,T)$ and thus $\bias(S,T)=d_\EM(S,T).$
\end{proof}

\begin{remark}
Here we point out a different proof of the fact that
$\bias_{\dtoDtv}(S,T)\le d_\EM(S,T)$ which does not
involve LP duality. Indeed $d_\EM(S,T)$ can be interpreted as giving the
cost of the best \emph{coupling} between the two distributions $S$ and $T$
subject to the penalty function $d(x,y).$ Recall, a coupling is a distribution
$(X,Y)$ over $V\times V$ such that the marginal distributions are $S$ and $T,$
respectively. The cost of the coupling is $\E d(X,Y).$
It is not difficult to argue directly that any such coupling gives an upper
bound on $\bias_{\dtoDtv}(S,T).$ We chose the linear programming proof
since it leads to additional insight into the tightness of the theorem.
\end{remark}

The situation for $\bias_{\dtoDinfty}$ is somewhat more complicated and we
do not get a tight characterization in terms of an Earthmover distance. We do
however have the following upper bound.
\begin{lemma}
\begin{equation}
\bias_{\dtoDinfty}(S,T)\le \bias_{\dtoDtv}(S,T)
\end{equation}
\end{lemma}
\begin{proof}
By Lemma~\ref{lem:tv2inf}, we have $D_\tv(\mu_x,\mu_y)\le
D_\infty(\mu_x,\mu_y)$ for any two distributions $\mu_x,\mu_y.$
Hence, every $(\dtoDinfty)$-Lipschitz mapping is
also $(\dtoDtv)$-Lipschitz. Therefore, $\bias_{\dtoDtv}(S,T)$ is a
relaxation of $\bias_{\dtoDinfty}(S,T).$
\end{proof}
\begin{corollary}
\begin{equation}
\bias_{\dtoDinfty}(S,T)\le d_\EM(S,T)
\end{equation}
\end{corollary}
For completeness we note the dual linear program obtained from the definition
of~$\bias_{\dtoDinfty}(S,T):$

\begin{align}
\bias_{\dtoDinfty}(S,T)
=
\quad\min\quad & \sum_{x\in V} \epsilon_x \notag \\
\label{eq:f}
\text{subject to}\quad
& \sum_{y\in V}f(x,y)+\epsilon_x
\ge \sum_{y\in V}f(y,x)e^{d(x,y)}+S(x)-T(x) \\
\label{eq:g}
& \sum_{y\in V}g(x,y)+\epsilon_x
\ge \sum_{y\in V}g(y,x)e^{d(x,y)} \\
& f(x,y),g(x,y)\ge0 \notag
\end{align}
Similar to the proof of Theorem~\ref{thm:em-tv}, we may interpret this program as
a \emph{flow} problem.  The variables $f(x,y), g(x,y)$ represent a
nonnegative \emph{flow} from $x$ to $y$ and
$\epsilon_x$ are slack variables. Note that the variables $\epsilon_x$ are
unrestricted as they correspond to an equality constraint. The first
constraint requires that $x$ has at least $S(x)-T(x)$ outgoing units of
flow in~$f.$ The
RHS of the constraints states that the penalty for \emph{receiving} a unit of
flow from~$y$ is $e^{d(x,y)}.$ However, it is no longer clear that we can
get rid of the variables $\epsilon_x,g(x,y).$

\begin{open}
Can we achieve a tight characterization of when $(D_\infty,d)$-Lipschitz
implies statistical parity?
\end{open}

\section{Fair Affirmative Action}
\label{sec:preferential}


In this section, we explore how to implement what may be called \emph{fair
affirmative action}. Indeed, a typical question when we discuss fairness
is, \emph{``What if we want to ensure statistical parity between two groups
$S$ and $T,$ but members of $S$ are less likely to be ``qualified''?}
In
Section~\ref{sec:lip-sp}, we have seen that when $S$ and $T$  are
``similar" then the Lipschitz condition implies statistical parity. Here we
consider the complementary case where $S$ and  $T$ are very different and
imposing statistical parity corresponds to preferential  treatment.
This is a cardinal question, which we examine with
a concrete example illustrated in Figure~\ref{fig:example1}.

\begin{figure}[h]
\begin{center}
\includegraphics[scale=0.75]{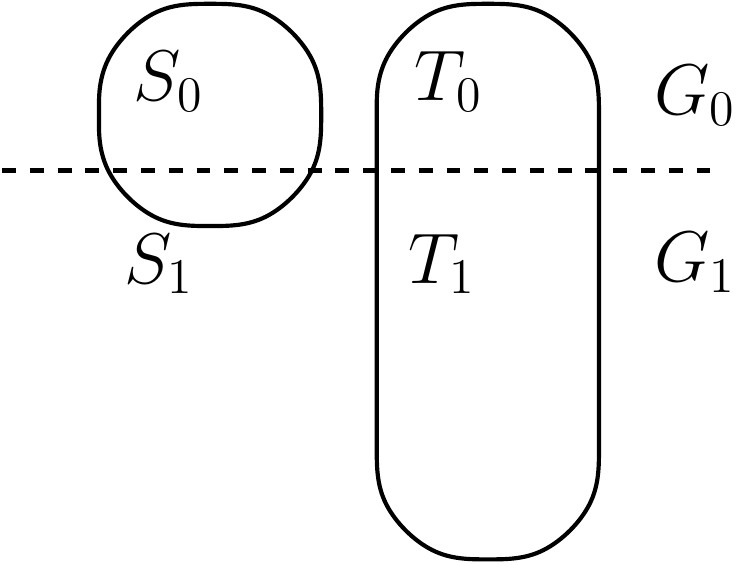}
\end{center}
\caption{$S_0=G_0\cap S$, $T_0=G_0\cap T$}
\label{fig:example1}
\end{figure}

For simplicity, let $T = \Scomp$.
Assume $|S|/|T \cup S| = 1/10$,
so $S$ is only $10\%$ of the population.
Suppose that our task-specific metric partitions $S \cup T$ into two groups,
call them $G_0$ and $G_1$, where
members of $G_i$ are very close to one another and very far from
all members of $G_{1-i}$.  Let $S_i$, respectively $T_i$, denote the
intersection $S \cap G_i$, respectively $T \cap G_i$, for $i = 0,1$.
Finally, assume $|S_0| = |T_0| = 9|S|/10$.
Thus, $G_0$ contains less than $20\%$ of the total population, and is equally
divided between $S$ and $T$.

The Lipschitz condition requires that members of each $G_i$ be treated
similarly to one another, but there is no requirement that members of
$G_0$ be treated similarly to members of $G_1$.  The treatment of
members of $S$, on average, may therefore be very different from the
treatment, on average, of members of $T$, since members of $S$ are
over-represented in $G_0$ and under-represented in $G_1$. Thus the
Lipschitz condition says nothing about statistical parity in this
case.

Suppose the members of $G_i$ are to be shown an advertisement $\ad_i$
for a loan offering, where the terms in $\ad_1$ are superior to those in
$\ad_0$.  Suppose further that the distance metric has partitioned the
population according to (something correlated with) credit score, with
those in $G_1$ having higher scores than those in $G_0$.

On the one hand, this seems fair: people with better ability to repay
are being shown a more attractive product. Now we ask two questions:
``What is the effect of imposing statistical parity?'' and
``What is the effect of failing to impose statistical parity?''

\paragraph{Imposing Statistical Parity.}
Essentially all of $S$ is in $G_0$, so for simplicity let us suppose
that indeed $S_0 = S \subset G_0$.  In this case, to ensure that members
of $S$ have comparable chance of seeing $\ad_1$ as do members of $T$,
members of $S$ must be treated, for the most part, like those in
$T_1$.  In addition, by the Lipschitz condition, members of $T_0$ must
be treated like members of $S_0 = S$, so these, also, are treated like
$T_1$, and the space essentially collapses, leaving only trivial
solutions such as assigning a fixed probability distribution on the
advertisements $(\ad_0, \ad_1)$ and showing ads according to this
distribution to each individual, or showing all individuals $\ad_i$
for some fixed~$i$.  However, while fair (all individuals are treated
identically), these solutions fail to take the \vendor 's loss
function into account.

\paragraph{Failing to Impose Statistical Parity.}
The demographics of the groups $G_i$ differ from the demographics of the
general population.  Even though half the individuals shown $\ad_0$ are
members of $S$ and half are members of $T$, this in turn can cause a problem
with fairness: an ``anti-$S$'' \vendor{} can effectively eliminate most
members of $S$ by replacing the ``reasonable'' advertisement $\ad_0$ offering
less good terms, with a blatantly hostile message designed to drive away
customers.  This eliminates essentially all business with members of $S$,
while keeping intact most business with members of $T$.  Thus, if members of
$S$ are relatively far from the members of $T$ according to the distance
metric, then satisfying the Lipschitz condition may fail to prevent some of
the unfair practices.

\subsection{An alternative optimization problem}
\label{sec:airport}

With the above discussion in mind, we now suggest a different approach, in
which we insist on statistical parity, but we relax the Lipschitz condition
between elements of $S$ and elements of $\Scomp$. This is consistent with the
essence of preferential treatment, which implies that elements in $S$ are
treated differently than elements in $T$.  The approach is inspired by the use
of the Earthmover relaxation in the context of metric labeling and
0-extension~\cite{KleinbergTa02,ChekuriKNZ04}.  Relaxing the $S \times T$
Lipschitz constraints also makes sense if the information about the distances
between members of $S$ and members of $T$ is of lower quality, or less
reliable, than the internal distance information within these two sets.

We proceed in two steps:
\begin{enumerate}
\item
\begin{enumerate}
\item
First we compute a mapping from elements in $S$ to distributions over
$T$ which transports the uniform distribution over $S$ to the uniform
distribution over $T$, while minimizing the total distance traveled.
Additionally the mapping preserves the Lipschitz condition between
elements within $S.$
\item
This mapping
gives us the following new loss function for elements of~$T$:
For $y \in T$ and $a \in A$ we define a new loss, $L'(y,a)$, as
\[
L'(y,a) = \sum_{x \in S} \mu_x(y)L(x,a) + L(y,a)\mcom
\]
where $\{\mu_x\}_{x\in S}$ denotes the mapping computed in step (a).
$L'$ can be viewed as a reweighting of the loss function $L$,
taking into account the loss on $S$ (indirectly through its mapping to $T$).

\end{enumerate}
\item Run the Fairness~LP only on $T$, using the new loss function~$L'$.
\end{enumerate}
Composing these two steps yields a
a mapping from $V=S\cup T$ into~$A.$

\medskip
\noindent
Formally, we can express the first step of this alternative approach
as a restricted
Earthmover problem defined as
\begin{align}
\label{eq:EM+L}
d_{\EM+\mathrm{L}}(S,T)\defeq\quad\min\quad
& \E_{x\in S}\E_{y\sim\mu_x} d(x,y)
\\
\text{subject to}\quad
& D(\mu_x,\mu_{x'}) \le d(x,x')   
\quad \text{for all}\quad x,x'\in S \notag\\
& D_\tv(\mu_S,U_T)\le\epsilon  \notag\\
& \mu_x\in \Delta(T) \quad \text{for all}\quad{x \in S}
\notag
\end{align}
Here, $U_T$ denotes the uniform distribution over~$T.$
Given $\{\mu_x\}_{x\in S}$ which
minimizes~(\ref{eq:EM+L}) and $\{\nu_x\}_{x\in T}$ which minimizes the
original fairness LP~(\ref{eq:fairness-lp}) restricted to $T,$
we define the mapping
$M\colon V\to\Delta(A)$ by putting
\begin{equation}\label{eq:composition}
M(x)=
\begin{cases}
\nu_x & x\in T\\
\E_{y\sim\mu_x}\nu_y & x\in S
\end{cases}\mper
\end{equation}

Before stating properties of the mapping $M$ we make some
remarks.

\begin{enumerate}
\item
Fundamentally, this new approach shifts from
minimizing loss, subject to the Lipschitz constraints,
to minimizing loss and disruption of $S \times T$ Lipschitz requirement,
subject to the parity and $S \times S$ and $T \times T$ Lipschitz constraints.
This gives us a bicriteria optimization problem, with a wide range of
options.

\item
We also have some flexibility even in the current version.  For
example, we can eliminate the re-weighting, prohibiting the \vendor{}
from expressing any opinion about the fate of elements in $S$.  This
makes sense in several settings.  For example, the \vendor{} may {\em
  request} this due to ignorance ({\it e.g.}, lack of market research)
about~$S$, or the \vendor{} may have some (hypothetical) special legal
status based on past discrimination against~$S$.

\item
It is instructive to compare the alternative approach
to a modification of the Fairness~LP in which
we enforce statistical parity and eliminate the Lipschitz
requirement on $S \times T$.
The alternative approach is more faithful to the $S \times T$ distances,
providing protection against the self-fulfilling
prophecy discussed in the Introduction, in which the  \vendor{}
deliberately selects the ``wrong'' subset of~$S$ while still
maintaining statistical parity.

\item
A related approach to addressing preferential treatment involves {\em
  adjusting} the metric in such a way that the Lipschitz condition
will imply statistical parity.  This coincides with at least one philosophy
behind affirmative action: that the metric does not fully reflect
potential that may be undeveloped
because of unequal access to resources.  Therefore, when we consider
one of the strongest individuals in $S$, affirmative action suggests
it is more appropriate to consider this individual as similar to one
of the strongest individuals of $T$ (rather than to an individual of
$T$ which is close according to the original distance metric).  In
this case, it is natural to adjust the distances between elements in
$S$ and $T$ rather than inside each one of the populations (other than
possibly re-scaling). This gives rise to a family of optimization
problems:
\begin{quote}
Find a new distance metric $d'$ which ``best approximates" $d$ under
the condition that $S$ and $T$ have small Earthmover distance under
$d'$,
\end{quote}
where we have the flexibility of choosing the measure
of quality to how well $d'$ approximates $d$.
\end{enumerate}

\noindent
Let $M$ be the mapping of Equation~\ref{eq:composition}.  The
following properties of~$M$ are easy to verify.
\begin{proposition}\label{pro:lip}
The mapping $M$ defined in~(\ref{eq:composition})
satisfies
\begin{enumerate}
\item statistical parity between $S$ and $T$ up to bias~$\epsilon,$
\item the Lipschitz condition for every pair $(x,y)\in (S\times
S)\cup(T\times T).$
\end{enumerate}
\end{proposition}
\begin{proof}
The first property follows since
\[
D_\tv(M(S),M(T))=D_\tv\left(\E_{x\in S}\E_{y\sim\mu_x}\nu_y,\E_{x\in
T}\nu_x\right)
\le D_\tv(\mu_S,U_T)\le\epsilon.
\]
The second claim is trivial for $(x,y)\in T\times T.$ So, let $(x,y)\in S\times
S.$ Then,
\[
D(M(x),M(y))
\le D(\mu_x,\mu_y)
\le d(x,y)\mper
\]
\end{proof}

We have given up the Lipschitz condition between $S$ and $T$, instead
relying on the terms $d(x,y)$ in the objective function to
discourage mapping $x$ to distant $y$'s. It turns out that the Lipschitz
condition between elements $x\in S$ and $y\in T$ is still maintained on
average and that the expected violation is given by $d_{\EM+\mathrm{L}}(S,T)$
as shown next.
\begin{proposition}
Suppose $D=D_\tv$ in~(\ref{eq:EM+L}). Then, the resulting mapping $M$
satisfies
\[
\E_{x\in S}
\max_{y\in T}
\Big[
D_\tv(M(x),M(y))- d(x,y)
\Big]\le
d_{\EM+\mathrm{L}}(S,T)\mper
\]
\end{proposition}
\begin{proof}
For every $x\in S$ and $y\in T$ we have
\begin{align*}
D_\tv(M(x),M(y))
& = D_\tv\left(\E_{z\sim\mu_x}M(z),M(y)\right)\\
&\le \E_{z\sim\mu_x}D_\tv\left(M(z),M(y)\right)
\tag{by Fact~\ref{fac:tvconv}}\\
&\le \E_{z\sim\mu_x}d(z,y)
\tag{Proposition~\ref{pro:lip} since $z,y\in T$}\\
&\le d(x,y)+\E_{z\sim\mu_x} d(x,z) \tag{by triangle inequalities}
\end{align*}
The proof is completed by taking the expectation over $x\in S.$
\end{proof}

An interesting challenge for future work is handling preferential
treatment of multiple protected subsets that are {\it not} mutually
disjoint. The case of disjoint subsets seems easier and in particular
amenable to our approach.


\section{Small loss in bounded doubling dimension}
\label{sec:expmechanism}

The general LP shows that given an instance $\cI$, it is possible to find an
``optimally fair'' mapping in polynomial time. The result however does not
give a concrete quantitative bound on the resulting loss.  Further, when the
instance is very large, it is desirable to come up with more efficient methods
to define the mapping.

We now give a fairness mechanism for which we can prove a bound on the loss
that it achieves in a natural setting. Moreover, the mechanism is
significantly more efficient than the general linear program.
Our mechanism is based on the exponential mechanism~\cite{McSherryTa07}, first
considered in the context of differential privacy.

We will describe the method in the natural setting where the mapping $M$ maps
elements of $V$ to distributions over $V$ itself. The method could be
generalized to a different set~$A$ as long as we also have a distance function
defined over $A$ and some distance preserving embedding of $V$ into $A$.  A
natural loss function to minimize in the setting where $V$ is mapped into
distributions over $V$ is given by the metric~$d$ itself. In this setting we
will give an explicit Lipschitz mapping and show that under natural
assumptions on the metric space $(V,d)$ the mapping achieves small loss.

\begin{definition}
Given a metric $d\colon V\times V\to\R$ the exponential mechanism
$\ExM\colon V\to\Delta(V)$ is defined by putting
\[
\ExM(x)\defeq [Z_x^{-1}e^{-d(x,y)}]_{y\in V}\mcom
\]
where $Z_x=\sum_{y\in V}e^{-d(x,y)}\mper$
\end{definition}

\begin{lemma}[\cite{McSherryTa07}]
The exponential mechanism is $(\dtoDinfty)$-Lipschitz.
\end{lemma}

One cannot in general expect the exponential mechanism to achieve small loss.
However, this turns out to be true in the case where $(V,d)$ has small
\emph{doubling dimension}. It is important to note that in differential
privacy, the space of databases does \emph{not} have small doubling dimension.
The situation in fairness is quite different. Many metric spaces arising in
machine learning applications do have bounded doubling dimension. Hence the
theorem that we are about to prove applies in many natural settings.

\begin{definition}
The \emph{doubling dimension} of a metric space $(V,d)$ is the smallest
number~$k$ such that for every $x\in V$ and every $R\ge0$ the ball of radius
$R$ around $x,$ denoted $B(x,R)=\{y\in V\colon d(x,y)\le R\}$ can be covered
by $2^k$ balls of radius $R/2.$
\end{definition}

We will also need that points in the metric space are not too close together.

\begin{definition}
We call a metric space~$(V,d)$ \emph{well separated} if there is a positive
constant $\epsilon>0$ such that $|B(x,\epsilon)|=1$ for all $x\in V.$
\end{definition}

\begin{theorem}\label{thm:expmech}
Let $d$ be a well separated metric space of bounded doubling dimension.
Then the exponential mechanism satisfies
\[
\E_{x\in V}\E_{y\sim \ExM(x)}d(x,y)= O(1)\mper
\]
\end{theorem}

\begin{proof}
Suppose $d$ has doubling dimension~$k.$
It was shown in~\cite{ChanGu08} that doubling dimension~$k$ implies for every
$R\ge0$ that
\begin{equation}\label{eq:ball-avg}
\E_{x\in V} |B(x,2R)|
\le 2^{k'}\E_{x\in V} |B(x,R)|\mcom
\end{equation}
where $k'=O(k).$
It follows from this condition and the assumption on $(V,d)$
that for some positive $\epsilon>0,$
\begin{equation}\label{eq:ball-1}
\E_{x\in V}|B(x,1)|\le
\left(\frac1\epsilon\right)^{k'}\E_{x\in V}|B(x,\epsilon)|= 2^{O(k)}\mper
\end{equation}
Then,
\begin{align*}
\E_{x\in V}\E_{y\sim \ExM(x)}  d(x,y)
& \le 1 +\E_{x\in V}\int_1^\infty \frac{re^{-r}}{Z_x}|B(x,r)|\rd r\\
& \le 1 +\E_{x\in V}\int_1^\infty re^{-r}|B(x,r)|\rd r \tag{since $Z_x\ge
e^{-d(x,x)}=1$}\\
& = 1 + \int_1^\infty re^{-r} \E_{x\in V}|B(x,r)|\rd r\\
& \le 1 + \int_1^\infty re^{-r} r^{k'}\E_{x\in V}|B(x,1)|\rd r
\tag{using~(\ref{eq:ball-1})}\\
& \le 1+ 2^{O(k)}\int_0^\infty r^{k'+1}e^{-r} \rd r\\
& \le 1+ 2^{O(k)}(k'+2)!\mper
\end{align*}
As we assumed that $k=O(1),$ we conclude
\[
\E_{x\in V}\E_{y\sim \ExM(x)} d(x,y)
\le 2^{O(k)}(k'+2)! \le O(1)\mper
\]
\end{proof}
\begin{remark}\label{rem:not-separated}
If $(V,d)$ is not well-separated, then for every constant $\epsilon>0,$ it
must contain a well-separated subset
$V'\subseteq V$ such that every point $x\in V$ has a neighbor $x'\in V'$ such
that $d(x,x')\le\epsilon.$ A Lipschitz mapping $M'$ defined on $V'$ naturally
extends to all of~$V$ by putting $M(x)=M'(x')$
where $x'$ is the nearest neighbor of $x$ in $V'.$
It is easy to see that the expected loss of $M$ is only an additive
$\epsilon$ worse than that of $M'.$ Similarly, the Lipschitz condition
deteriorates by an additive~$2\epsilon,$ i.e., $D_\infty(M(x),M(y))\le
d(x,y)+2\epsilon\mper$ Indeed, denoting the nearest neighbors in $V'$
of $x,y$ by $x',y'$ respectively, we have
$D_\infty(M(x),M(y))=D_\infty(M'(x'),M'(y'))\le d(x',y')\le
d(x,y)+d(x,x')+d(y,y')\le d(x,y)+2\epsilon.$ Here, we used the triangle
inequality.
\end{remark}
The proof of Theorem~\ref{thm:expmech} shows an exponential dependence on the doubling
dimension~$k$ of the underlying space in the error of the
exponential mechanism.
The next theorem shows that the loss of any Lipschitz mapping has
to scale at least linearly with $k.$ The proof follows from a packing argument
similar to that in~\cite{HardtTa10}.
The argument is slightly
complicated by the fact
that we need to give a lower bound on the \emph{average} error
(over $x\in V$) of any mechanism.
\begin{definition}
A set $B\subseteq V$ is called an \emph{$R$-packing} if $d(x,y)>R$
for all $x,y\in B.$
\end{definition}
Here we give a lower bound using a metric space that may not be
well-separated. However, following Remark~\ref{rem:not-separated},
this also shows that any mapping defined on a well-separated subset of the
metric space must have large error up to a small additive loss.
\begin{theorem}\label{thm:lb}
For every~$k\ge2$ and every large enough
$n\ge n_0(k)$ there exists an  $n$-point metric space of
doubling dimension $O(k)$ such that
any $(\dtoDinfty)$-Lipschitz mapping~$M\colon V\to\Delta(V)$ must satisfy
\[
\E_{x\in V}\E_{y\sim M(x)}d(x,y)\ge \Omega(k)\mper
\]
\end{theorem}
\begin{proof}
Construct~$V$ by randomly picking $n$ points from a $r$-dimensional sphere
of radius $100k.$ We will choose $n$ sufficiently large and $r=O(k).$
Endow $V$ with the Euclidean distance~$d.$ Since
$V\subseteq\mathbb{R}^r$ and $r=O(k)$
it follows from a well-known fact that the doubling dimension of $(V,d)$
is bounded by~$O(k).$

\begin{claim}\label{cla:random}
Let $X$ be the distribution obtained by choosing a random $x\in V$ and
outputting a random $y\in B(x,k).$ Then, for sufficiently large $n,$ the
distribution $X$ has statistical distance at most $1/100$ from the uniform
distribution over $V.$
\end{claim}
\begin{proof}
The claim follows from standard arguments showing that for large enough $n$
every point $y\in V$ is contained in approximately equally many balls of
radius~$k.$
\end{proof}

Let $M$ denote any $(\dtoDinfty)$-Lipschitz mapping and denote its
error on a point $x\in V$ by
\[
R(x) =\E_{y\sim M(x)}d(x,y)\mper
\]
and put $R= \E_{x\in V}R(x).$
Let $G=\{x\in V\colon R(x)\le 2R\}.$
By Markov's inequality $|G|\ge n/2.$

Now, pick $x\in V$ uniformly at random
and choose a set $P_x$ of $2^{2k}$ random points
(with replacement) from $B(x,k).$
For sufficiently large dimension $r=O(k),$ it follows from concentration of
measure on the sphere that $P_x$ forms a $k/2$-packing with probability,
say, $1/10.$

Moreover, by
Claim~\ref{cla:random}, for random $x\in V$ and random $y\in B(x,k),$ the
probability that $y\in G$ is at least $|G|/|V| - 1/100\ge 1/3.$ Hence, with
high probability,
\begin{equation}
|P_x\cap G|\ge 2^{2k}/10\mper
\end{equation}

Now, suppose $M$ satisfies $R\le k/100.$ We will lead this to a contradiction
thus showing that $M$ has average error at least $k/100.$ Indeed, under the
assumption that $R\le k/100,$ we have that for every $y\in G,$
\begin{equation}
\Pr\left\{ M(y)\in B(y,k/50)\right\}\ge\frac12\mcom
\end{equation}
and therefore
\begin{align*}
1\ge \Pr\left\{M(x)\in \cup_{y\in P_x\cap G}B(y,k/2)\right\}
& = \sum_{y\in P_x\cap G}\Pr\left\{M(x)\in B(y,k/2)\right\} \tag{since $P_x$ is
a $k/2$-packing}\\
& \ge \sum_{y\in P_x\cap G}\exp(-k)\Pr(M(y)\in B(y,k/2)) \tag{by the Lipschitz
condition}\\
& = \frac{2^{2k}}{10}\cdot \frac{\exp(-k)}2
 > 1\mper
\end{align*}
This is a contradiction which shows that $R>k/100.$
\end{proof}

\begin{open}
Can we improve the exponential dependence on the doubling dimension in our
upper bound?
\end{open}


\section{Discussion and Future Directions}
\label{sec:future}

In this paper we introduced a framework for characterizing fairness in
classification.  The key element in this framework is a requirement that
similar people be treated similarly in the classification.  We developed an
optimization approach which balanced these similarity constraints with a
vendor's loss function.  and analyzed when this local fairness condition
implies statistical parity, a strong notion of equal treatment.  We also
presented an alternative formulation enforcing statistical parity, which is
especially useful to allow preferential treatment of individuals from some
group.  We remark that although we have focused on using the metric as a
method of defining and enforcing fairness, one can also use our approach to
{\em certify} fairness (or to detect unfairness).  This permits us to evaluate
classifiers even when fairness is defined based on data that simply isn't
available to the classification algorithm\footnote{This observation is due to
Boaz Barak.}.

Below we consider some open questions and directions for future work.

\subsection{On the Similarity Metric}
\label{sec:metric}
As noted above, one of the most challenging aspects of our work is justifying
the availability of a distance metric.  We argue here that the notion of a
metric already exists in many classification problems, and we consider some
approaches to building such a metric.

\subsubsection{Defining a metric on individuals}

The imposition of a metric already occurs in many
classification processes. Examples include credit scores\footnote{We remark
that the credit score is a one-dimensional metric that suggests an obvious
interpretation as a measure of quality rather than a measure of similarity.
When the metric is defined over multiple attributes such an interpretation is
no longer clear.} for loan applications,
and combinations of test scores and grades for some college admissions.
In some cases, for reasons of social engineering, metrics may be
adjusted based on membership in various groups, for example, to
increase geographic and ethnic diversity.

The construction of a suitable metric can be partially
automated using existing machine learning techniques. This is true in
particular for distances $d(x,y)$ where $x$ and $y$ are both in the same
protected set or both in the general population. When comparing individuals
from different groups, we may need human insight and domain information.
This is discussed further
in Section~\ref{sec:building}.

Another direction, which intrigues us but which have not
yet pursued, is particularly relevant to the context of on-line
services (or advertising):
allow users to specify attributes they do or do not want
to have taken into account in classifying content of interest.
The risk, as noted early on in this work, is
that attributes may
have redundant encodings in other attributes,
including encodings of which the
user, the ad network, and the advertisers may all be unaware.
Our notion of fairness
can potentially give a refinement of the ``user empowerment" approach
by allowing a user to participate in defining the metric that is used
when providing services to this user (one can imagine for example a
menu of metrics each one supposed to protect some subset of
attributes). Further research into the feasibility of this approach is
needed, in particular, our discussion throughout
this paper has assumed that a single metric
is used across the board.  Can we make sense out of the idea of applying
different metrics to different users?

\subsubsection{Building a metric via metric labeling}
\label{sec:building}

One approach to building the metric is to first build a metric on $\Scomp$,
say, using techniques from machine learning, and then ``inject''
members of $S$ into the metric by mapping them to members of $S$ in a
fashion consistent with observed information.
In our case, this observed information would come from the human
insight and domain information mentioned above.
Formally, this can be captured by the problem
of {\em metric labeling}~\cite{KleinbergTa02}:
we have a collection of $|\Scomp|$ {\em labels} for which a metric is
defined, together with $|S|$ objects, each of which is to be assigned a label.

It may be expensive to access this extra information
needed for metric labeling. We may ask the question of how much information
do we need in order to approximate the result we would get were we to
have all this information.  This is related to
our next question.

\subsubsection{How much information is needed?}
\label{sec:howmuch}
Suppose there is an unknown metric $d^*$ (the right metric) that we are trying
to find. We can ask an expert panel to tell us $d^*(x,y)$ given $(x,y)\in V^2.$
The experts are costly and we are trying to minimize the number of calls we need
to make. The question is: How many queries $q$ do we need to make to be able to compute
a metric $d\colon V\times V\to\mathbb{R}$ such that the distortion between $d$
and $d^*$ is at most $C,$ i.e.,
\begin{equation}
\sup_{x,y\in V}
\max\left\{\frac{d(x,y)}{d^*(x,y)},\frac{d^*(x,y)}{d(x,y)}\right\}\le C\mper
\end{equation}


The problem can be seen as a variant of the well-studied question of
constructing \emph{spanners}.  A spanner is a small implicit
representation of a metric $d^*.$ While this is not exactly what we want, it
seems that certain spanner constructions work in our setting as well,
are willing to relax the embedding problem by permitting a certain fraction of
the embedded edges to have arbitrary distortion, as any finite metric can
be embedded, with constant slack and constant distortion,
into constant-dimensional Euclidean space~\cite{AbrahamBCDGKNS05}.

\subsection{Case Study on Applications in Health Care}
An interesting direction for a case study is suggested by
another Wall Street Journal article (11/19/2010) that describes
the (currently experimental) practice of
insurance risk assessment via online tracking.  For example,
food purchases and exercise habits correlate with certain diseases.
This is a stimulating, albeit alarming, development.  In the most
individual-friendly interpretation described
in the article, this provides a method for assessing risk
that is faster and less expensive than the current practice of testing blood and
urine samples.
``Deloitte and the life insurers stress the databases wouldn't be used
to make final decisions about applicants. Rather, the process would
simply speed up applications from people who look like good
risks. Other people would go through the traditional assessment
process.''~\cite{WSJ:insurance}
Nonetheless, there are risks to the insurers, and preventing
discrimination based on protected status should therefore be of interest:
\begin{quote}
``The information sold by marketing-database firms is lightly regulated. But
  using it in the life-insurance application process would ``raise questions''
  about whether the data would be subject to the federal Fair Credit Reporting
  Act, says Rebecca Kuehn of the Federal Trade Commission's division of
  privacy and identity protection. The law's provisions kick in when ``adverse
  action'' is taken against a person, such as a decision to deny insurance or
  increase rates.''
\end{quote}

As mentioned in the introduction, the AALIM project~\cite{AALIM} provides
similarity information suitable for the health care setting. While their work
is currently restricted to the area of cardiology, future work may extend to
other medical domains. Such similarity information may be used to
assemble a metric that decides which individual have similar medical
conditions. Our framework could then employ this metric to ensure that similar
patients receive similar health care policies. This would help to address the
concerns articulated above. We pose it as an interesting direction for future
work to investigate how a suitable fairness metric could be extracted from the
AALIM system.

\subsection{Does Fairness Hide Information?}
We have already discussed the need for 
hiding (non-)membership in $S$
in ensuring fairness. We now ask a converse question:
Does fairness in the context of advertising hide information from the
advertiser?

Statistical parity has the interesting effect that it eliminates
\emph{redundant encodings} of $S$ in terms of $A,$ in the sense that
after applying $M,$ there is no $f\colon A\to\bits$ that can be biased
against $S$ in any way. This prevents certain attacks that aim to
determine membership in~$S$.

Unfortunately, this property is not hereditary.  Indeed, suppose that the
advertiser wishes to target HIV-positive people.  If the set of HIV-positive
people is protected, then the advertiser is stymied by the statistical parity
constraint.  However, suppose it so happens that the advertiser's utility
function is extremely high on people who are not only HIV-positive but who
also have AIDS.  Consider a mapping that satisfies statistical parity for
``HIV-positive,'' but also maximizes the advertiser's utility.  We expect that
the necessary error of such a mapping will be on members of
``HIV$\backslash$AIDS,'' that is, people who are HIV-positive but who do not
have AIDS.  In particular, we don't expect the mapping to satisfy statistical
parity for ``AIDS'' -- the fraction of people with AIDS seeing the
advertisement may be much higher than the fraction of people with AIDS in the
population as a whole. Hence, the advertiser can in fact target ``AIDS''.

Alternatively, suppose people with AIDS are mapped to a region $B \subset A$,
as is a $|{\rm AIDS}|/|\rm{HIV~positive}|$ fraction of HIV-negative
individuals.  Thus, being mapped to $B$ maintains statistical parity for the
set of HIV-positive individuals, meaning that the probability that a random
HIV-positive individual is mapped to $B$ is the same as the probability that a
random member of the whole population is mapped to $B$.  Assume further that
mappings to $A \setminus B$ also maintains parity.  Now the advertiser can
refuse to do business with all people with AIDS, sacrificing just a small amount of business
in the HIV-negative community.

These examples show that statistical parity is not a good 
method of hiding sensitive information in
targeted advertising.  A natural question, not yet pursued, is whether we can
get better protection using the Lipschitz property with a suitable metric.

\section*{Acknowledgments}

We would like to thank Amos Fiat for a long and wonderful discussion
which started this project.  We also thank Ittai Abraham, Boaz Barak,
Mike Hintze, Jon Kleinberg, Robi Krauthgamer, Deirdre Mulligan, Ofer
Neiman, Kobbi Nissim, Aaron Roth, and Tal Zarsky for helpful
discussions.  Finally, we are deeply grateful to Micah Altman for
bringing to our attention key philosophical and economics works.

\bibliographystyle{alpha}
\bibliography{awareness}

\newcommand{\etalchar}[1]{$^{#1}$}
\begin{thebibliography}{AAN{\etalchar{+}}98}

\bibitem[{AAL}]{AALIM}
{AALIM}.
\newblock {\tt http://www.almaden.ibm.com/cs/projects/aalim/}.

\bibitem[AAN{\etalchar{+}}98]{AjtaiANRSW98}
Miklos Ajtai, James Aspnes, Moni Naor, Yuval Rabani, Leonard~J. Schulman, and
  Orli Waarts.
\newblock Fairness in scheduling.
\newblock {\em Journal of Algorithms}, 29(2):306--357, November 1998.

\bibitem[ABC{\etalchar{+}}05]{AbrahamBCDGKNS05}
Ittai Abraham, Yair Bartal, Hubert T.-H. Chan, Kedar Dhamdhere, Anupam Gupta,
  Jon~M. Kleinberg, Ofer Neiman, and Aleksandrs Slivkins.
\newblock Metric embeddings with relaxed guarantees.
\newblock In {\em FOCS}, pages 83--100. IEEE, 2005.

\bibitem[BS06]{BansalS06}
Nikhil Bansal and Maxim Sviridenko.
\newblock The santa claus problem.
\newblock In {\em Proc.\ $38$th STOC}, pages 31--40. ACM, 2006.

\bibitem[Cal05]{Calsamiglia05}
Catarina Calsamiglia.
\newblock Decentralizing equality of opportunity and issues concerning the
  equality of educational opportunity, 2005.
\newblock Doctoral Dissertation, Yale University.

\bibitem[CG08]{ChanGu08}
T.-H.~Hubert Chan and Anupam Gupta.
\newblock Approximating {TSP} on metrics with bounded global growth.
\newblock In {\em Proc.\ $19$th Symposium on Discrete Algorithms (SODA)}, pages
  690--699. ACM-SIAM, 2008.

\bibitem[CKNZ04]{ChekuriKNZ04}
Chandra Chekuri, Sanjeev Khanna, Joseph Naor, and Leonid Zosin.
\newblock A linear programming formulation and approximation algorithms for the
  metric labeling problem.
\newblock {\em SIAM J. Discrete Math.}, 18(3):608--625, 2004.

\bibitem[DMNS06]{DworkMcNiSm06}
Cynthia Dwork, Frank McSherry, Kobbi Nissim, and Adam Smith.
\newblock Calibrating noise to sensitivity in private data analysis.
\newblock In {\em Proc.\ $3$rd TCC}, pages 265--284. Springer, 2006.

\bibitem[Dwo06]{Dwork06}
Cynthia Dwork.
\newblock Differential privacy.
\newblock In {\em Proc.\ $33$rd ICALP}, pages 1--12. Springer, 2006.

\bibitem[Fei08]{Feige08}
Uri Feige.
\newblock On allocations that maximize fairness.
\newblock In {\em Proc.\ $19$th Symposium on Discrete Algorithms (SODA)}, pages
  287--293. ACM-SIAM, 2008.

\bibitem[FT11]{FeigeT11}
Uri Feige and Moshe Tennenholtz.
\newblock Mechanism design with uncertain inputs (to err is human, to forgive
  divine).
\newblock In {\em Proc.\ $43$rd STOC}, pages 549--558. ACM, 2011.

\bibitem[HT10]{HardtTa10}
Moritz Hardt and Kunal Talwar.
\newblock On the geometry of differential privacy.
\newblock In {\em Proc.\ $42$nd STOC}. ACM, 2010.

\bibitem[Hun05]{Hunt05}
D.~Bradford Hunt.
\newblock Redlining.
\newblock Encyclopedia of Chicago, 2005.

\bibitem[JM09]{gaydar}
Carter Jernigan and Behram~F.T. Mistree.
\newblock Gaydar: Facebook friendships expose sexual orientation.
\newblock {\em First Monday}, 14(10), 2009.

\bibitem[KT02]{KleinbergTa02}
Jon~M. Kleinberg and {\'E}va Tardos.
\newblock Approximation algorithms for classification problems with pairwise
  relationships: metric labeling and markov random fields.
\newblock {\em Journal of the ACM (JACM)}, 49(5):616--639, 2002.

\bibitem[MT07]{McSherryTa07}
Frank McSherry and Kunal Talwar.
\newblock Mechanism design via differential privacy.
\newblock In {\em Proc.\ $48$th Foundations of Computer Science (FOCS)}, pages
  94--103. IEEE, 2007.

\bibitem[Rab93]{Rabin93}
M.~Rabin.
\newblock Incorporating fairness into game theory and economics.
\newblock {\em The American Economic Review}, 83:1281--1302, 1993.

\bibitem[Raw01]{Rawls01}
John Rawls.
\newblock {\em Justice as Fairness, A Restatement}.
\newblock Belknap Press, 2001.

\bibitem[SA10]{WSJ:8-4-10}
Emily Steel and Julia Angwin.
\newblock On the web's cutting edge, anonymity in name only.
\newblock {\em The Wall Street Journal}, 2010.

\bibitem[SM10]{WSJ:insurance}
Leslie Scism and Mark Maremont.
\newblock Insurers test data profiles to identify risky clients.
\newblock {\em The Wall Street Journal}, 2010.

\bibitem[You95]{Young95}
H.~Peyton Young.
\newblock {\em Equity}.
\newblock Princeton University Press, 1995.

\bibitem[Zar11]{Zarsky11}
Tal Zarsky.
\newblock Private communication.
\newblock 2011.

\end{thebibliography}

\appendix

\section{Catalog of Evils}
\label{sec:taxonomy}
We briefly summarize here behaviors against which we wish to protect.
We make no attempt to be formal.
Let $S$ be a protected set.
\begin{enumerate}
\item
{\em Blatant explicit discrimination.}  This is when membership in $S$
is explicitly tested for and a ``worse'' outcome is given to members of
$S$ than to members of $\Scomp$.
\item
{\em Discrimination Based on Redundant Encoding.}  Here the explicit
test for membership in $S$ is replaced by a test that is, in practice,
essentially equivalent.  This is a successful attack against
``fairness through blindness,'' in which
the idea is to simply ignore protected attributes such as \emph{sex} or
\emph{race}. However, when 
personalization and advertising decisions are based on
months or years of on-line activity, there is a very real possibility that
membership in a given demographic group is embedded holographically in the
history.  Simply deleting, say, the Facebook ``sex'' and ``Interested in
men/women'' bits almost surely does not hide homosexuality.  This point was
argued by the (somewhat informal) ``Gaydar'' study~\cite{gaydar} in which a
threshold was found for predicting, based on the sexual preferences of his
male friends, whether or not a given male is interested in men.
Such {\em redundant encodings} of sexual preference and other attributes need
not be explicitly known or recognized as such, and yet can still have a
discriminatory effect.

\item
{\em Redlining.}
A well-known form of discrimination based on redundant encoding.
The following definition appears in an article by~\cite{Hunt05}, which
contains the history of the term, the practice, and its consequences:
``Redlining is the practice of arbitrarily denying or limiting financial
services to specific neighborhoods, generally because its residents are people
of color or are poor.''
\item
{\em Cutting off business with a segment of the population in which
membership in the protected set is disproportionately high.}
A generalization of redlining, in which
members of $S$ need not be a majority of the redlined population;
instead, the fraction of the redlined population belonging to $S$ may
simply exceed the fraction of $S$ in the population as a whole.
\item
{\em \Sfp .}
Here the \vendor {}
advertiser is willing to cut off its nose to spite its face,
deliberately choosing the ``wrong'' members of $S$
in order to build a bad ``track record'' for $S$. A less malicious \vendor{}
may simply select \emph{random} members of~$S$ rather than qualified members,
thus inadvertently building a bad track record for~$S.$
\item
{\em Reverse tokenism.}
This concept arose in the context of imagining what might be a convincing
refutation to the claim ``The bank denied me a loan because I am a member
of~$S$.''  One possible refutation might be the exhibition of an ``obviously
more qualified'' member of $\Scomp$ who is also denied a loan.
This might be compelling, but by sacrificing one really good candidate
$c \in \Scomp$ the bank could refute all charges of discrimination
against $S$.  That is, $c$ is a token rejectee; hence the term
``reverse tokenism''  (``tokenism'' usually refers to accepting a token
member of $S$).
We remark that the general question of explaining decisions seems
quite difficult, a situation only made worse by the existence
of redundant encodings of attributes.
\end{enumerate}

\end{document}